\documentclass[default,iicol]{sn-jnl}

\usepackage{graphicx}%
\usepackage{multirow}%
\usepackage{amsmath,amssymb,amsfonts}%
\usepackage{amsthm}%
\usepackage{mathrsfs}%
\usepackage[title]{appendix}%
\usepackage{xcolor}%
\usepackage{textcomp}%
\usepackage{manyfoot}%
\usepackage{algorithm}%
\usepackage{algorithmicx}%
\usepackage{algpseudocode}%
\usepackage{listings}%
\usepackage{array}  

\catcode`\|=12\relax 
\usepackage{orcidlink} 

\newcommand{\R}{\mathcal{R}}

\DeclareMathOperator{\EX}{\mathbb{E}}
\DeclareMathOperator{\VAR}{\mathbb{V}}

\usepackage{xparse}

\newcolumntype{L}[1]{>{\raggedright\let\newline\\\arraybackslash\hspace{0pt}}m{#1}}
\newcolumntype{C}[1]{>{\centering\let\newline\\\arraybackslash\hspace{0pt}}m{#1}}
\newcolumntype{R}[1]{>{\raggedleft\let\newline\\\arraybackslash\hspace{0pt}}m{#1}}

\NewDocumentCommand{\codeword}{v}{%
\texttt{{#1}}%
}
\lstMakeShortInline[language=Matlab]|
\newcommand\inlineMatlab[2][]{\lstinline[language=Matlab,#1]+#2+}

\newcommand\MyBox[2]{
  \fbox{\lower0.75cm
    \vbox to 1.2cm{\vfil
      \hbox to 1.2cm{\hfil\parbox{1.2cm}{\centering #1\\#2}\hfil}
      \vfil}%
  }%
}

\newenvironment{confusionMat}[9]
{
\begin{table}[ht]
\centering
\renewcommand\arraystretch{1.5}
\setlength\tabcolsep{0pt}
\resizebox{.2\textwidth}{!}{
\begin{tabular}{c @{\hspace{1.1em}}>{\bfseries}r @{\hspace{.7em}}c @{\hspace{0.4em}}c}
  \multirow{9}{*}{\rotatebox{90}{\parbox{5cm}{\bfseries\centering Predicted Class}}} & 
    & \multicolumn{2}{c}{\bfseries True Class} \\
  & & \bfseries 0 & \bfseries 1 \\
  & 0 & \MyBox{#1}{(#2\%)} & \MyBox{#3}{(#4\%)} \\[3em]
  & 1 & \MyBox{#5}{(#6\%)} & \MyBox{#7}{(#8\%)}  \\[4.5em]
\end{tabular}}
\caption{#9}
\end{table}
}
{
}

\newtheorem{theorem}{Theorem}
\begin{document}

\title[Efficient Solution of Portfolio Optimization Problems via Dimension Reduction and Sparsification]{Efficient Solution of Portfolio Optimization Problems via Dimension Reduction and Sparsification}


\author*[1]{\fnm{Cassidy K.} \sur{Buhler} \orcidlink{0000-0003-4157-4273}}\email{cb3452@drexel.edu} 

\author[1]{\fnm{Hande Y.} \sur{Benson} \orcidlink{0000-0002-5554-9928}}\email{hvb22@drexel.edu}

\affil*[1]{\orgdiv{Department of Decision Sciences and MIS}, \orgname{Bennett S. LeBow College of Business, Drexel University}, \orgaddress{\street{3220 Market St}, \city{Philadelphia}, \postcode{19104}, \state{PA}, \country{USA}}}


\abstract{The Markowitz mean-variance portfolio optimization model aims to balance expected return and risk when investing. However, there is a significant limitation when solving large portfolio optimization problems efficiently: the large and dense covariance matrix. Since portfolio performance can be potentially improved by considering a wider range of investments, it is imperative to be able to solve large portfolio optimization problems efficiently, typically in microseconds. We propose dimension reduction and increased sparsity as remedies for the covariance matrix. The size reduction is based on predictions from machine learning techniques and the solution to a linear programming problem. We find that using the efficient frontier from the linear formulation is much better at predicting the assets on the Markowitz efficient frontier, compared to the predictions from neural networks. Reducing the covariance matrix based on these predictions decreases both runtime and total iterations. We also present a technique to sparsify the covariance matrix such that it preserves positive semi-definiteness, which improves runtime per iteration. The methods we discuss all achieved similar portfolio expected risk and return as we would obtain from  a full dense covariance matrix but with improved optimizer performance. }




\maketitle
\bmhead{Acknowledgments}

The authors would like to thank Drs. Christopher Gaffney and Matthew Schneider for their feedback on an earlier version of the paper.
\section{Introduction}\label{sec1}

The Markowitz mean-variance portfolio optimization model \cite{markowitz} aims to balance expected return and risk when investing. Investors with different risk tolerances can choose to put different levels of relative importance on these objectives and an efficient frontier can be constructed representing the optimal portfolios for all possible risk tolerances. 

Let $p_{t,j}$ represent the (known) closing price for stock $j = 1, \ldots, N$ on day $t = 1, \ldots, (T-1)$.  The return $x_{t,j}$ for stock $j = 1, \ldots, N$ on day $t = 2, \ldots, (T-1)$ is calculated as  
\begin{align}
x_{t,j} = \dfrac{ p_{t,j}-p_{t-1,t}}{p_{t-1,j}}.
\end{align}
For portfolio weights $w \in \R^N$, the portfolio return at time $t = 2, \ldots, (T-1)$ is computed by   
\begin{align}
R_t = \sum_{j=1}^N w_j x_{t,j}.
\end{align}
Denoting the return matrix as $X \in \R^{(T-2) \times N}$, we can also write $R = Xw$.

The portfolio return on day $T$, $\textbf{R}_T$, is
a random variable with 
\begin{align}
    \EX[\textbf{R}_T] = \mu^{T} w 
    , \qquad \VAR[\textbf{R}_T] =w^{T} \Sigma  w
\end{align}
where $\Sigma = \textnormal{cov}(X)$ and 
$\mu_j = \EX[x_{T,j}],\ j = 1,...,N$.

The Markowitz model is formulated as:

\begin{equation}
\begin{aligned}
\max_{w} \quad & \mu^{T}w - \lambda w^{T} \Sigma w\\
\textrm{s.t.} \quad &e^{T}w = 1\\
& w \geq 0
\label{mvmodel}
\end{aligned}
\end{equation} 
where $e$ is a ones vector of appropriate size and $\lambda \geq 0$ is the risk aversion parameter. As we vary $\lambda$, we obtain all optimal portfolios and represent them as the {\em efficient frontier} (Figure \ref{fig:ef}).

\begin{figure*}[t]
    \centering
    \includegraphics[width = 0.8\textwidth]{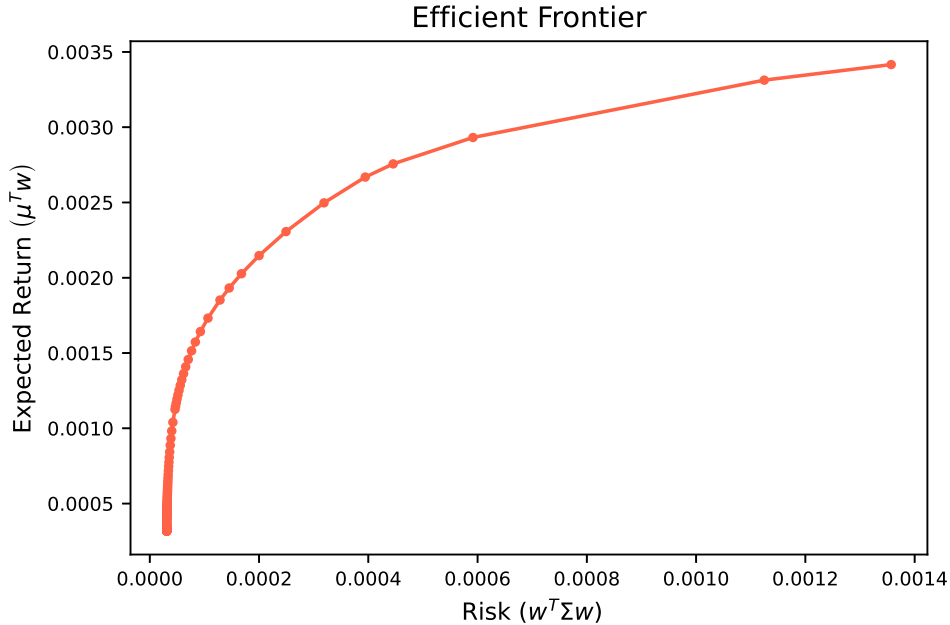}
    \caption{The efficient frontier for 374 stocks selected from the S\&P500 using the daily returns from January 23rd 2012 to December 31st 2019.}
    \label{fig:ef}
\end{figure*}

The expected return, $\mu$, is typically measured as the mean of historical returns in the Markowitz framework, but it can be replaced by any forecast of the returns. Risk is interpreted as volatility, and the covariance matrix, $\Sigma$, is used to capture the variance of returns for individual investments and to evaluate opportunities to mitigate (or increase) portfolio risk by simultaneously choosing investments with negatively (or positively) correlated returns.  While $\mu$ is typically updated daily (or as new forecasts are available), $\Sigma$ uses extensive historical information to assess correlation and is generally static.

There are other optimal portfolio selection strategies that attempt to improve on the Markowitz model and many of these models do not use a covariance matrix. However, we wish to focus primarily on the Markowitz model, as it is a simple framework and has many applications outside of investing, namely, in the energy sector \cite{dellano2016addressing, zhang2018optimal,arnesano2012extension,ostadi2020risk,kellner2019sustainability}. 

In addition, advancements in machine learning and other data science techniques have improved the ability to more accurately calculate $\mu$ and $\Sigma$, and are studied in recent literature \cite{ban2018machine,mulvey2017machine, wang2019applying,bennett2014data}.


As investment platforms grow, it becomes important to calculate these quantities and develop optimal portfolios in real time.  Since portfolio performance can be improved by considering a wider range of investments, it is imperative to solve large instances of \eqref{mvmodel} efficiently.  

The biggest challenge to efficiency is the use of $\Sigma$: it is large and dense. Nevertheless, it has several advantages we wish to exploit: it is a good representation of hedging opportunities, it is positive semidefinite (thus, \eqref{mvmodel} is a convex quadratic optimization problem), and it does not need to be updated frequently. 

For large instances, the number of stocks included in the portfolios on the efficient frontier is often significantly smaller than $N$. This property gives an opportunity to reduce or sparsify the covariance matrix by omitting entries that do not impact the optimal risk. 
 \begin{figure*}[t]
    \centering
    \includegraphics[width = 0.99\textwidth]{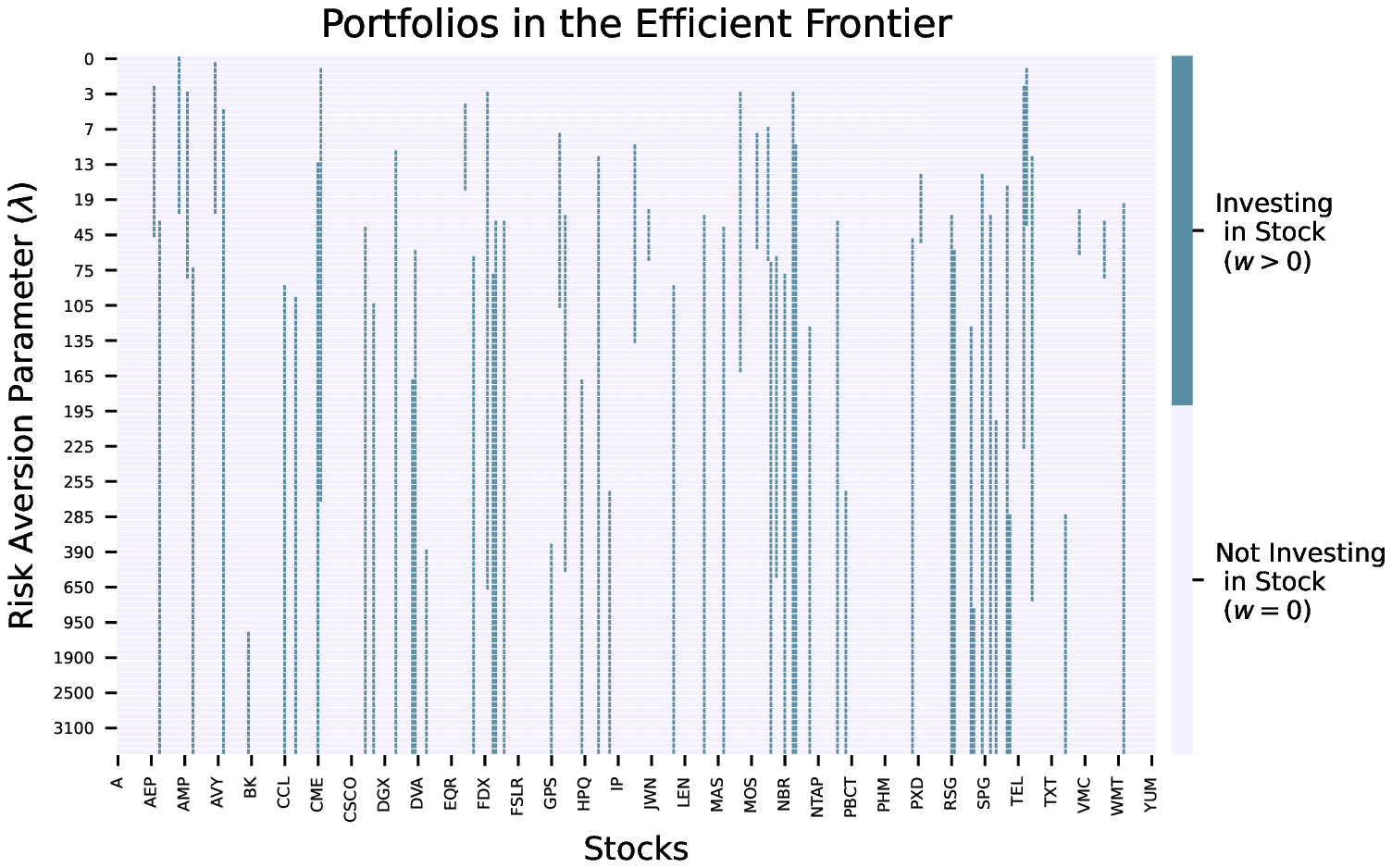}
    \caption{119 portfolios form the efficient frontier for the problem instance from Figure \ref{fig:ef}. Each row is an optimal portfolio for a given $\lambda$. A stock is teal if it is included in the portfolio, and light purple otherwise.}
    \label{fig:heatmap}
\end{figure*}
For example, in Figure \eqref{fig:heatmap}, only 64 out of the 374 stocks were included in any optimal portfolio. This implies that dimensions or the number of nonzeros in $\Sigma$ can be significantly reduced.

Therefore, in order to improve our efficiency, we will focus on {\em reducing the size or increasing the sparsity of $\Sigma$}.
\begin{itemize}
    \item The proposed size reduction techniques are predictive: machine learning algorithms, namely neural networks, or reformulation of \eqref{mvmodel} as a linear programming problem (LP) is used to predict the assets in the optimal portfolios along the efficient frontier.  The risk term in \eqref{mvmodel} is then limited to these assets.  For each approach, our hypothesis is that the predictive models will give a sufficient match to the investments in the optimal portfolios. 
    \item The sparsification techniques are based on the correlation matrix of the stock returns.  Correlations close to -1 or 1 not only represent strong and sustained relationships among pairs of stocks but they also represent significant enough contributions to the risk term of \eqref{mvmodel}.  As such, correlations close to zero can be replaced by zero, which yields a sparse matrix.  Care must be taken to retain positive semidefiniteness.  The main hypothesis for this technique is that the variance and large correlations are sufficient to determine the optimal portfolios and represent their risk.
\end{itemize}

We will refer to the three techniques as reduction by neural networks, reduction by LP, and sparsification by correlation and now present the motivation for choosing these three techniques for our study:
\begin{enumerate}
    \item {\em Reduction by neural networks}: Neural networks--once trained--are much faster and can yield better results than other probabilistic classifiers such as logistic regression or naive Bayes.  They have been used in literature for stock price prediction  \cite{solin2019forecasting, freitas2009prediction} or classifying stocks based on market performance  \cite{hargreaves2013stock, fu2018machine}. Existing literature is focused on predicting the data for \eqref{mvmodel}, whereas our goal in this paper is to predict which stocks appear in the optimal basis for \eqref{mvmodel} for any value of $\lambda$.  Neural networks are ideally suited for such tasks.
    \item {\em Reduction by LP}: Another way to predict the optimal solution is to solve a similar, simpler optimization problem. By redefining the way that risk is measured, \cite{vanderbei2020linear} outlines how to reformulate \eqref{mvmodel} as an LP. Not only can the LP be solved faster than the QP for each value of $\lambda$, the use of the parametric simplex method, as proposed by \cite{vanderbei2020linear}, allows us to recover the entire efficient frontier within the solution of a single LP. 
    \item {\em Sparsification by correlation}: \eqref{mvmodel} is typically solved using an interior-point method, which requires the factorization of a dense matrix in the KKT system at each iteration.  Moreover, the nonnegativity requirement means that the dense matrix changes in each iteration and requires $\mathcal{O}\left(N^3\right)$ operations. On the other hand, the time complexity of sparse factorization is proportional to the number of nonzero elements in the matrix \cite{gilbert1992sparse}, which can yield significant improvement.
\end{enumerate}
It is important to answer the question of why sparsification may be pursued when we know that a size reduction will typically improve runtimes.  To illustrate the motivation behind sparsification, we generated a large sparse matrix and a small dense matrix to compare how long it takes to factor each matrix. The first matrix is 99\% sparse and $10000 \times 10000$, and the second matrix is dense and $2000 \times 2000$. We used Cholesky factorization and factored the matrices 100 times. The large sparse matrix took an average of 0.0037 seconds to factor while the small dense matrix took 0.0752 seconds.  It is, therefore, possible that we can do better than size reduction techniques when a significant level of sparsity can be attained.

The outline of the paper is as follows.  In Section \ref{sec2}, we introduce our reduction methods with a long short-term memory neural network and a linear programming formulation, followed by the proposed method of sparsification by correlation.  In Section \ref{sec3}, we provide details on our financial data, which consists of daily closing prices for individual stocks in the S\&P 500 index.  The numerical results are provided in Section \ref{sec4} for all three proposed methods, and we account for other commonly used predictive methods (logistic regression, naive Bayes, and pattern recognition neural networks) as well.  We finish with a discussion of our findings and future work in Section \ref{sec5}.

\section{Methods}\label{sec2}

In this section, we present the three proposed methods in detail.

\subsection{Reduction by Neural Networks: Prediction with Long Short-Term Memory Network}\label{sec2sub1}

We used classification to predict the stocks that will be included in the portfolio and use these predictions to create a reduced covariance matrix. The network architecture we chose to implement is a long short-term memory (LSTM) network \cite{hochreiter1997long} due to its high predictive ability for financial data as shown in literature \cite{fischer2018deep,ta2020portfolio}. Note that predicting the solution to every optimal portfolio on the efficient frontier is a much more complex problem than predicting stock prices or returns, but the network architecture is quite similar.

In order to train a network, we split up the matrix of historical daily returns $X$ by rows to get a train and test set. The first $t$ rows of $X$ are denoted as $X_{\text{train}}$ and the remaining $T-t$ rows as $X_{\text{test}}$. We then use $X_{\text{train}}$ as the input to the Markowitz model \eqref{mvmodel} and obtain a solution $w^*(\lambda)$ for every $\lambda$. We introduce the target vector $y_{\text{train}}^*$, defined as
\begin{align}
    (y_{\text{train}})_i^* =
    \begin{cases}
    0 & w_i^*(\lambda)= 0 \ \text{for all } \lambda \geq 0\\
    1 & \text{otherwise}.\\
    \end{cases}
\end{align}
This information is used to classify each stock into two groups: If $(y_{\text{train}})_i^* = 0$, stock $i$ is said to be in Class 0, and if $(y_{\text{train}})_i^*$, stock $i$ is said to be in Class 1.  We can then repeat this same process with $X_{\text{test}}$ to obtain $y_{\text{test}}^*$. 

We train an LSTM network with $X_{\text{train}}$ and $y_{\text{train}}^*$ and the network gives a prediction which we call $\Tilde{y}$. This prediction is then used to reduce the full covariance matrix $\Sigma$ to $\Tilde{\Sigma}$ by removing from $\Sigma$ row $i$ and column $i$ for each $i$ with $\Tilde{y_i} = 0$.

LSTM units learn from five weights, and they have three gates that control how much information gets through: forget gate, input gate, and output gate. They also combine each new input with old hidden output, and they can carry some of this older data. This cell state remembers how much information from the past based on the gates. Our network has 5 layers:  
\begin{enumerate}
    \item \textbf{Sequence Input Layer:} The single dimension input layer is given $X_{\text{train}}$.
    \item \textbf{LSTM Layer:} We elected to use 150 hidden nodes.  The state activation function is the hyperbolic tangent function. The gate activation function is the sigmoid function. The input and recurrent weights are initialized with Matlab's default Glorot \cite{glorot2010understanding} and Orthogonal \cite{saxe2013exact}, respectively.
    \item \textbf{Fully Connected Layer:} This layer takes in the output from the previous layer and reshapes the data to prepare for the classification.
    \item \textbf{Softmax Layer:} This function is defined in the appendix, computed with $X_{\text{train}}$ and the target vector $y_{\text{train}}^*$.
    \item \textbf{Weighted Cross Entropy Classification Output:} The weights are computed by the priors of training data and the function is also defined in the appendix. 
\end{enumerate}
After we train the network based on $X_{\text{test}}$, we input $X_{\text{test}}$ and the network gives a prediction for $\Tilde{y}$. We use this prediction to compare with $y_{\text{test}}^*$ to gain insight on the predictive ability of the network. 

\subsection{Reduction by LP: Prediction with the Parametric Simplex Method}\label{sec2sub2} One interpretation of the risk aversion parameter $\lambda$ in \eqref{mvmodel} is that it is the penalty parameter in a two-objective model.  In such a framework, it is also possible to shift the role of the penalty term to the portfolio return.  To do so, we introduce a new penalty parameter $\gamma \ge 0$.

We now review the LP reformulation of \eqref{mvmodel} as presented in \cite{vanderbei2020linear}.  Recall that the covariance matrix can be written as 
\[
\Sigma = A^T A, \text{ where } 
A = X - \bar{X}, 
\bar{X}_{ij} = \frac{1}{T-1} \sum_{t=1}^{T-1} X_{t,j}. 
\]
The matrix $A$ represents the deviation of the actual return from its expectation, and the covariance term in the Markowitz model measures the magnitude of this deviation in a quadratic sense.  To form the LP, we instead measure it in an absolute sense:
\[
\begin{array}{ll}
\max_{w} \quad & \gamma \mu^{T}w - \| A^T w \|_1 \\
\textrm{s.t.} \quad &e^{T}w = 1\\
& w \geq 0.
\end{array}
\]
Then, we introduce an auxiliary variable $v$ to complete the reformulation: 
\begin{equation}
\begin{aligned}
\max_{w} \quad & \gamma \mu^{T}w - \frac{1}{T}\sum_{i=1}^T v_i\\
\textrm{s.t.} \quad &-v \leq \textbf{A}^{T}w \leq v\\
& e^Tw = 1\\
& w,v \geq 0.\\
\label{lp}
\end{aligned}
\end{equation}
As described in \cite{vanderbei2020linear}, this problem is then solved using the parametric self-dual simplex method (\cite{dantzig1998linear}) with a specific pivot sequence that uncovers optimal portfolios for every value of $\gamma \ge 0$ in one pass of the method.  Our proposal in this paper is to use these portfolios to predict which investments will be included in the solutions of \eqref{mvmodel} and reduce its risk term accordingly. More specifically, solving the model using training data and then using the predictions to reduce the testing data.  

\subsection{Sparsification by Correlation}\label{sec2sub3}
In this method, we will replace the covariance matrix, $\Sigma$, in \eqref{mvmodel} with a sparse matrix obtained by replacing entries corresponding to small correlations with zero.  There are two reasons why this approach is reasonable: First, the covariance matrix is positive semidefinite, we expect it to be diagonally dominated. Therefore, the typical decision of whether or not to include an investment in the optimal portfolio is made first and foremost with a focus on return and variance. The off-diagonal values impact the optimal solution only if they have a significant impact on the risk. Second, large entries in the covariance matrix can be an indicator of strong correlation between investments’ historical returns.  We would expect strong correlations to be exhibited for many time periods, making it quite likely that such a relationship will continue into the future. As such, it stands to reason that the small values in the correlation matrix can be replaced with zeros, and that this change would lead to the corresponding covariance matrix to be sparse.

The challenge here is two-fold: (1) we need to determine an appropriate definition of a small correlation, and  (2) the resulting sparse covariance matrix needs to be positive semidefinite.  To address these challenges, we use a two-step process.  We determine a threshold for correlation, below which the correlation and, therefore, the covariance is replaced by 0.  After obtaining the resulting sparse matrix, we add back in some of the original values to re-establish positive semidefiniteness.  The domain of possible threshold values is taken to be the entries of the correlation matrix.

Let $\theta$ be the dense $N \times N$ correlation matrix, $\tau$ be a threshold value (with $0 \le \tau < 1$)
obtained from the unique values of $\theta$, and $\hat{\Sigma}(\tau)$ be the sparse $N \times N$ covariance matrix defined for threshold $\tau$. We propose the following simple scheme to sparsify $\Sigma$: 
\begin{equation} \label{reset}
    \hat{\Sigma}_{ij}(\tau) = 
    \begin{cases} 
    0 &-\tau \leq \theta_{ij} \leq \tau \\
    \Sigma_{ij} &\text{otherwise} \\ 
    \end{cases}, \text{ for all } i, j = 1,\dots,N 
\end{equation}
However, $\hat{\Sigma}(\tau_k)$ may be indefinite for any value of $\tau_k$ with no clear pattern, especially as the matrix size increases.  In order to remedy the indefiniteness, we pair our scheme with partial matrix completion.  Before formally describing our algorithm, we illustrate it with an example here (provided with further details in the Appendix).

Consider a $4 \times 4$ sparse covariance matrix whose nonzero elements are labeled with stars: 
\[
  \begin{bmatrix}
* & & &  \\
 & *& *&  \\
 & *& *&  *\\
 & & * & * \\
  \end{bmatrix}
\]
Since there are off-diagonal values in columns 2, 3 and 4, the matrix completion method would add back in the elements $(2,4)$ and $(4,2)$.  The resulting matrix would consist of the variance of the returns of Stock 1 and the covariance matrix of the returns for Stocks 2, 3, and 4: 
\[
 \begin{bmatrix}
 * &  & &  \\
 & *& *&  *\\
 & *& *&  *\\
 & * & * & * \\
  \end{bmatrix}
\]
This method is shown as Algorithm \ref{algo1}.

\begin{algorithm}
 \caption{Partial Matrix-Completion for Threshold $\tau$}\label{algo1}
\begin{algorithmic}[1]
\State set $\hat{\Sigma}(\tau)$ according to Equation \eqref{reset}.
\State $\hat{n} \gets$ column(s) where there exists an off-diagonal non-zero entry in $\hat{\Sigma}$ 
\For{$i \in \hat{n}$ and $j \in \hat{n}$}
\State $\hat{\Sigma}_{ij} (\tau) = \Sigma_{ij}$
\EndFor
\end{algorithmic}
\end{algorithm}

We can show that Algorithm \ref{algo1} always yields a positive semidefinite matrix, thereby ensuring that replacing $\Sigma$ with $\hat{\Sigma}$ in \eqref{mvmodel} always yields a convex quadratic programming problem.

\begin{theorem}
Let $\Sigma$ be a covariance matrix for $X$.  Then, $\hat{\Sigma}(\tau)$ obtained by applying Algorithm \ref{algo1} to $\Sigma$ is positive semidefinite for any value of $\tau$.
\end{theorem}

\begin{proof}
Since $\Sigma$ is a covariance matrix, it is symmetric and positive semidefinite.  Let S denote the set of column indices of $\hat{\Sigma}(\tau)$ with only diagonal entries and let $S'$ denote its remaining column indices.  Without loss of generality, we can denote the form of $\hat{\Sigma}(\tau)$ as
\[
    \hat{\Sigma}(\tau) = 
    \begin{bmatrix}
    C & \\
    & D \\
    \end{bmatrix},
\]
where the submatrix $C$ is a diagonal matrix whose entry $C_{j,j}$ is the variance of Stock $j$'s returns for each stock $j \in S$ and submatrix $D$ is a full matrix whose entries match the corresponding terms in $\Sigma$.  Therefore, $D$ can also be defined as the covariance matrix of the returns of the stocks in set $S'$, and, as such, it is positive semidefinite.  
We can define $C^{\frac{1}{2}}$ as the diagonal matrix with the standard deviations of the stock returns from $S$ and $D^{\frac{1}{2}}$ as the corresponding entries of $A$.  Then, we can rewrite 
\[
    \hat{\Sigma}(\tau) = 
    \begin{bmatrix}
    C^{\frac{1}{2}} & \\
    & D^{\frac{1}{2}} \\
    \end{bmatrix}^T
    \begin{bmatrix}
    C^{\frac{1}{2}} & \\
    & D^{\frac{1}{2}} \\
    \end{bmatrix}.
\]
Therefore, $\hat{\Sigma}(\tau)$ is positive semidefinite.
\end{proof}

\section{Data}\label{sec3}
The financial data was collected from Yahoo! Finance over the dates January 23rd 2012 to December 31st 2019 for 374 firms selected from the S\&P500 index. The returns are computed using the percentage change of the closing price. This gives a total of 1999 days of data.  For reduction by neural networks, we partition the data into the first 1499 days for training and the subsequent 500 days for testing. We chose a 75/25 split since it would give a large enough testing set such that the covariance matrix would not be rank deficient.

Modifying the risk aversion parameter $\lambda$ can be challenging, since we cannot use sensitivity analysis within interior-point methods to guarantee that we find every portfolio along the efficient frontier. We approached this by finding the $\lambda$ that achieves the minimum risk of a portfolio, where the minimum risk is the optimal objective to the following problem.
\begin{equation}
\begin{aligned}
\min_{w} \quad &  w^{T} \Sigma w\\
\textrm{s.t.} \quad &e^{T}w = 1\\
& w \geq 0
\label{minRisk}
\end{aligned}
\end{equation} 
We then increase the $\lambda$ in (4) until we obtain a portfolio risk that is within the distance $\epsilon = 10^{-8}$ of its minimum risk. This is the maximum $\lambda$ for the dataset. Then, we used the findings of \cite{BS05b} on warmstarts and knowledge of prior performance of interior-point methods on portfolio optimization problems to start $\lambda$ at 0 and increase it slowly until we reached its maximum value.  If any problem took more than 15 iterations to solve for a value of $\lambda$, we determined that we had increased $\lambda$ too fast and tried a smaller value first.

\section{Numerical Testing}\label{sec4}
Numerical testing was conducted to compare the optimizer performance and portfolio performance of each approach. 
We used the {\sc yfinance} package \cite{yfinance} v0.1.54 to obtain stock data for data collection. The numerical analysis was performed on a computer with a 2.7 GHz Quad-Core Intel Core i7 processor. We used Gurobi Optimizer version 10.0.02 build v10.0.2rc0 (mac64) in {\sc Matlab} 2023a to solve \eqref{mvmodel}. Gurobi was also used to solve \eqref{lp} for an initial value of $\gamma$, followed by pivots conducted in a Python 3 implementation of the parametric simplex method. The neural networks were implemented with {\sc Matlab} Deep Learning toolbox.



For every solution on the efficient frontier, we recorded the CPU time, total number of iterations, and CPU time per iteration (TPI). We used 632 values of $\lambda$ to construct the efficient frontiers. To represent the values of $\lambda$ slowly increasing, we used a collection of linearly spaced vectors. This is given in Table \ref{tab:lambda}.  Given that CPU time can vary on each run for reasons external to the numerical testing, we computed the efficient frontier 20 times and recorded the mean for the CPU times. 
\begin{table}[ht]
    \centering
    \begin{tabular}{lll}
    \toprule
        \inlineMatlab{x} & \inlineMatlab{y} & \inlineMatlab{n}  \\
        \botrule 
        0 & 1 & 121  \\
        1 & 2 & 61 \\
        2 & 4 & 51\\
        4 & 7 & 61\\
        7& 20 & 61\\
        21 & 60 & 81\\
        60 & 100 & 51\\
        100 & 200 & 71\\
        200 & \inlineMatlab{maxLambda} & 81\\
        \botrule
    \end{tabular}
    \caption{In {\sc Matlab} syntax, \inlineMatlab{linspace(x,y,n)} generates a vector of \inlineMatlab{n} evenly spaced points between \inlineMatlab{x} and \inlineMatlab{y}. $\lambda$ is the union of these vectors in the table where \inlineMatlab{maxLambda} is given by the process outlined in Section \ref{sec3}. }
    \label{tab:lambda}
\end{table}


In addition, each solution on the efficient frontier gives an optimal objective, expected risk, expected return and actual return which we also reported. The portfolio performance looks at the mean of the objective function $\mu^{T}w - \lambda  w^T \Sigma w$,  expected risk $\VAR[\textbf{R}] = w^T \Sigma w$, expected return $\EX[\textbf{R}] = \mu^T w$, and actual return. The actual portfolio return ($\Tilde{\textbf{R}}$) is computed using the actual stock returns on day $T$ ($\Tilde{\mu}$):
\[ \Tilde{\textbf{R}} = \Tilde{\mu}^T w. \]
Lastly, we recorded the number of assets which have nonzero weights in any efficient portfolio.  The use of the primal simplex solver in Gurobi allowed for a precise count of the number of nonzero weights in each solution.

\subsection{Reductions by Neural Networks and by Linear Programming}\label{sec4sub1}
As presented in Section \ref{sec2sub1}, we used the adaptive moment estimation (Adam) optimizer \cite{kingma2014adam} as the training solver method, with the parameter settings given in Table \ref{tab:hyperparameters}. The portfolio and optimizer results are found in Table \ref{tab:portfolioReductive} and Table \ref{tab:optimizerReductive}, respectively. 
\begin{table}[ht]
\renewcommand{\arraystretch}{1.25}
\centering
\begin{tabular}{@{}ll@{}}
\toprule
Hyperparameter & Value \\ 
\midrule
  Gradient Decay Factor  &  0.9000 \\
  Squared Gradient Decay Factor &  0.9990\\
  Initial Learning Rate     & 0.00025\\
  L2 Regularization &  0.00001 \\ 
  Max Epochs & 60 \\
  Mini Batch Size & 187\\
  Number of Mini Batches & 2 \\
  Epsilon Denominator Offset &  1.0000e-08\\
\botrule
\end{tabular}
\caption{Hyperparameters for the Adam optimizer. }
\label{tab:hyperparameters}
\end{table}




\begin{table}[ht]
\renewcommand{\arraystretch}{1.25}
\centering
\begin{tabular}{@{}lrrr@{}}
\toprule
Reduction Type & None & LSTM & LP  \\
\botrule
Matrix Size & $374 \times 374$ & $128 \times 128$ & $82 \times 82$ \\
\midrule
Total Iter & 168 & 71 & 95 \\ 
\midrule
Avg Iter & 0.2658 & 0.1123 & 0.1503\\
\midrule
Avg Time    &  0.0183  & 0.0162  & 0.0165\\ 
\midrule
Avg TPI   & 0.0688 &  0.1439 &  0.1097\\ 
\botrule
\end{tabular}
\caption{Optimizer Performance for Reductive Methods.}
\label{tab:optimizerReductive}
\end{table}
\begin{table}[t]
\renewcommand{\arraystretch}{1.25}

\centering
\begin{tabular}{@{}lrrr@{}}
\toprule
Reduction Type & None & LSTM & LP  \\
\botrule
Matrix Size & $374 \times 374$ & $128 \times 128$ & $82 \times 82$ \\
\midrule
Assets             & 65        & 43 &  30 \\
\midrule
 Obj     &    -0.00229 & -0.00344 & -0.00250 \\
\midrule
 $\VAR[\textbf{R}]$ & 0.00029  & 0.00017  & 0.00017  \\
\midrule
 $\EX[\textbf{R}] $     & 0.00161  & 0.00145  & 0.00128  \\
\midrule
 $\Tilde{\textbf{R}}$   &0.01490 &	0.00104	& 0.00379\\
\botrule
\end{tabular}
\caption{Portfolio Performance for Reductive Methods.}
\label{tab:portfolioReductive}
\end{table}

For each predictive method, we include a confusion matrix \cite{stehman1997selecting} to visualize the accuracy in Table 5 and Table 6.  
           
\confusionMat{246}{65.8}{46}{12.3}{63}{16.8}{19}{5.1}{The confusion matrix for the LP network predictions shows a 70.9\% overall accuracy for the 374 stocks.  The classifications are 0 or 1, where 0 represents not investing in the stock and 1 is investing in the stock. }

\confusionMat{205}{54.8}{41}{11}{104}{27.8}{24}{6.4}{The confusion matrix for the LSTM network predictions shows a 61.2\% overall accuracy for the 374 stocks.  The classifications are 0 or 1, where 0 represents not investing in the stock and 1 is investing in the stock. }

\subsection{Sparsification by Correlation}\label{sec4sub2}

Our proposed method of using correlations to sparsify the covariance matrix, followed by partial matrix completion to recover positive semidefiniteness, was effective and efficient in our numerical experiments.  We measure sparsity by the number of zero-valued elements divided by the total number of elements:
\begin{align}
 \text{Sparsity} = \frac{\text{number of zero-valued elements}}{\text{total number of elements}}
\end{align}

\begin{figure*}[ht]
    \centering    
    \includegraphics[width = 0.6\textwidth]{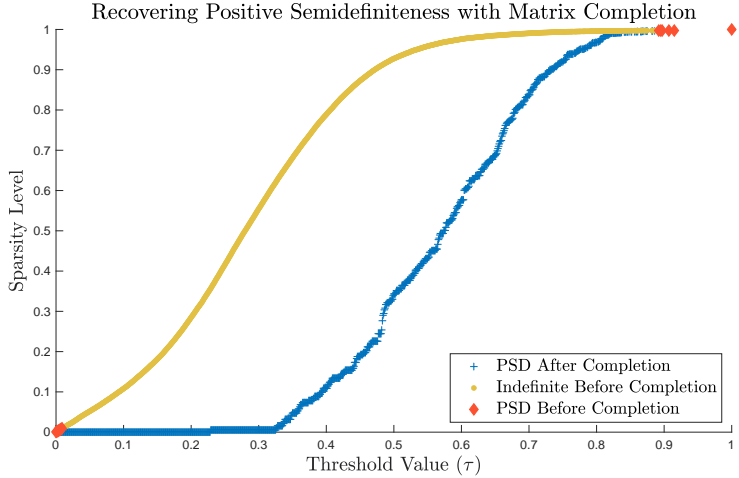}
    \label{fig:matrixComp}
    \caption{This graph compares the positive definiteness and sparsity using the brute force search with and without matrix completion. }
\end{figure*}
We see that for $0.4 \leq \tau \leq 0.7$, partial matrix completion as a means to re-establish positive semidefiniteness also requires a significant compromise on sparsity. However, a correlation coefficient is considered to have a weak relationship for $0 \leq \tau \leq 0.5$ and moderate relationship for $0.5 \leq \tau \leq 0.8$  \cite{devore2012modern}. Therefore, using Algorithm \ref{algo1} with $\tau > 0.8$ can indeed result in the retention of only strong relationships in our risk measure while simultaneously promoting both sparsity and positive semidefiniteness.  

To measure the method's optimizer and portfolio performance, we chose seven positive semidefinite matrices with varying sparsity levels where the 0\% sparse matrix is the full covariance matrix. The input to the Markowitz model are the covariance matrices and the full $\mu$ vector, which we did not sparsify since it is a linear coefficient thus has little strain on the runtime. 

\begin{table}[ht]
\renewcommand{\arraystretch}{1.25}
\setlength{\tabcolsep}{4pt} 
\centering
\resizebox{.47\textwidth}{!}{
\begin{tabular}{@{}L{1.4cm}rrrrrrr@{}}
\toprule
Sparsity  & 0\%  & 50\% & 60\% & 70\% & 80\% & 90\% & 99\% \\ 
\botrule
Total Iter & 168 & 347 & 343 & 382 & 419 & 453 & 540 \\ 
\midrule
Avg Iter     & 0.266&   0.549  &   0.543  &   0.604  &   0.663  &   0.717  &   0.854\\ 
\midrule
Avg Time & 0.018 & 0.018 &  0.018 &  0.017  &0.017 & 0.017   & 0.017\\ 
\midrule
Avg TPI   &  0.069   &  0.033 &  0.032   &0.029  & 0.026 & 0.023  & 0.019\\ 
\botrule
\end{tabular}}
\caption{Optimizer Performance for Sparsification by Correlation. }
\end{table}

\begin{table}[ht]
\renewcommand{\arraystretch}{1.3}
\setlength{\tabcolsep}{4pt} 
\centering
\resizebox{.47\textwidth}{!}{
\begin{tabular}{@{}L{1.07cm}rrrrrrr@{}}
\toprule
Sparsity             & 0\%  & 50\% & 60\% & 70\% & 80\% & 90\% & 99\% \\ 
\botrule
Assets  & 65       & 150      & 184      & 212      & 245      & 286      & 356\\
\midrule
Obj & -0.002  & -0.009 & -0.010 & -0.011 & -0.012 & -0.013 & -0.015  \\
\midrule
$\VAR[\textbf{R}]$ & 0.0003  & 0.0003  & 0.0003  & 0.0003  & 0.0003  & 0.0003  & 0.0003  \\
\midrule
$\EX[\textbf{R}] $  & 0.0016  & 0.0018  & 0.0018  & 0.0018  & 0.0018  & 0.0018  & 0.0018  \\

\midrule
 $\Tilde{\textbf{R}}$ & 0.0149  & 0.0182  & 0.0183  & 0.0190  & 0.0194  & 0.0198  & 0.0208  \\

\botrule
\end{tabular}
}
\caption{Portfolio Performance for Sparsification by Correlation.}
\end{table}

In the dataset for our numerical testing discussed so far (details provided in Section \ref{sec3}), we favored having a large dataset in order to observe and establish strong correlations.  However, given the strong performance of the S\&P500 index during this time period and the fact that we only collected data on stocks that were represented in the index for the entire time horizon, we did not observe any significant negative correlations in the dataset.  Nevertheless, negative correlations, and the resulting negative covariances, are a critical element of reducing risk via hedging opportunities.  As such, we repeated the above experiment with a shorter dataset that has more stocks: 497 assets for 502 days starting November 21 2017 to November 19th 2019. The covariance matrix for this dataset included strong negative correlations.  

\begin{table}[ht]
\renewcommand{\arraystretch}{1.25}
\setlength{\tabcolsep}{4pt} 
\centering
\resizebox{.47\textwidth}{!}{
\begin{tabular}{@{}L{1.4cm}rrrrrrr@{}}
\toprule
Sparsity  & 0\%  & 50\% & 60\% & 70\% & 80\% & 90\% & 99\% \\
\botrule
Total Iter & 188 & 347 & 392 & 436 & 511 & 563 & 660 \\ 
\midrule
Avg Iter & 0.298 & 0.549  & 0.620 &   0.69 &   0.809&  0.891 &   1.044 \\ 
\midrule
Avg Time & 0.022 &  0.018 &  0.019 &  0.018 &   0.018 &  0.017 &  0.017\\ 
\midrule
Avg TPI &  0.073 &   0.034 &   0.03 &  0.026  & 0.022 & 0.019  &  0.016\\ 
\botrule
\end{tabular}}
\caption{Optimizer Performance for Sparsification by Correlation when Negative Correlations are Present.}
\end{table}

\begin{table}[ht]
\renewcommand{\arraystretch}{1.25}
\setlength{\tabcolsep}{4pt} 
\centering
\resizebox{.47\textwidth}{!}{
\begin{tabular}{@{}L{1.07cm}rrrrrrr@{}}
\toprule
Sparsity  & 0\%  & 50\% & 60\% & 70\% & 80\% & 90\% & 99\% \\
\botrule
Assets   & 72     & 191    & 228    & 266    & 316    & 371    & 465    \\
\midrule
 Obj  & -0.002 & -0.011 & -0.013 & -0.013 & -0.014 & -0.016 & -0.017 \\
\midrule
 $\VAR[\textbf{R}]$ & 0.0003 & 0.0003 & 0.0003 & 0.0003 & 0.0003 & 0.0003 & 0.0003 \\
\midrule
 $\EX[\textbf{R}] $    & 0.0015 & 0.0018 & 0.0018 & 0.0018 & 0.0018 & 0.0018 & 0.0018 \\
\midrule
 $\Tilde{\textbf{R}}$ & 0.0109 & 0.0095 & 0.0101 & 0.0100 & 0.0097 & 0.0102 & 0.0105 \\
\botrule
\end{tabular}}
\caption{Portfolio Performance for Sparsification by Correlation when Negative Correlations are Present}
\end{table} 
\section{Conclusion}\label{sec5}
Reducing the size of the covariance matrix decreased the number of iterations and CPU runtime, yet increased the runtime per iteration. The dimension reduction improved $\VAR[\textbf{R}]$, yet has lower $\EX[\textbf{R}]$ and $\Tilde{\textbf{R}}$. Both reduction and no reduction showed that $\Tilde{\textbf{R}}$ differs from $\EX[\textbf{R}]$, which is not uncommon since the Markowitz model has been documented to emphasize estimation error, especially in larger portfolios \cite{frankfurter1971portfolio, jobson1980estimation}. 

A common issue with the machine learning methods, is that the imbalanced dataset made it challenging for models to distinguish when an asset should be invested in, because Class 1 is much smaller than Class 0. In both the training and testing sets, 17.38\% is Class 1 and 82.62\% is Class 0. This results in the model labeling nearly every asset as Class 0, thus being impractical as the model does not invest in any stocks. We saw this primarily in Section \ref{sec4sub2}, so we opted to take extra precautions, such as using the weighted cross entropy. This weight would impose a penalty, thus forcing Class 1 to be larger. The LSTM network predicted similar Class 1 and Class 2 priors, but struggled to distinguish between the two classes. It follows that LP had better predictions compared to LSTM with 70.9\% and 61.2\% overall accuracy, respectively.  

Moreover, the reduction methods required time to train models before predicting on testing data. The LSTM network took 26 minutes to train, while the LP model took an hour. Although the training is a one time occurrence, it is worth considering this cost.  

We conclude that our proposed LSTM model provided insight on the predictive abilities of neural networks in portfolio optimization, while motivating future extensions. In our model, the hyperparameters were kept simple to avoid overfitting, and we hope to explore alternative network architecture to better manage imbalanced data.


Sparsification most notably decreased the runtime per iteration. As the covariance matrix grows more sparse, the optimal solution invests in more assets.  This requires more basis changes to compute, which explains the increase in total number of iterations and average runtime. With respect to the portfolio performance, the $\VAR[\textbf{R}]$ was similar to the dense portfolio and $\EX[\textbf{R}]$ increased as the sparsity increased. This behavior could suggest that removing the off-diagonal values from the covariance matrix may increase diversification in the case of predominantly positively correlated stock returns.

Overall, sparsification yields better optimizer results while achieving similar portfolio $\VAR[\textbf{R}]$ and $\EX[\textbf{R}]$. In our testing, a highly efficient solver such as Gurobi exhibited improvement, thus we would expect the improvement to be even more significant for larger problem sizes. As future work, we will investigate other sparsification procedures that preserve positive semidefiniteness. In addition, we plan to develop a parametric self-dual simplex method for quadratic programming, with specific application to the Markowitz model for portfolio optimization.

\backmatter

\section*{Declarations}

\subsection{Funding}
The authors declare that no funds, grants, or other support were received during the preparation of this manuscript.

\subsection{Competing Interests}
The authors have no relevant financial or non-financial interests to disclose.



\subsection{Data}
The financial data was collected from Yahoo! Finance and is available at the following repository \url{https://github.com/cassiebuhler/PODS}.

\subsection{Code availability}
The code is available at \url{https://github.com/cassiebuhler/PODS}.

\subsection{Authors' contributions}

All authors contributed to the study conception and design. Material preparation, data collection and analysis were performed by Cassidy Buhler. The first draft of the manuscript was written by Cassidy Buhler and all authors commented on previous versions of the manuscript. All authors read and approved the final manuscript.


\noindent
\bigskip

\begin{appendices}

\section{Activation and Loss Functions}\label{secA1}
The activation and loss functions for the LSTM are provided in Table \ref{tab:activationloss}.
\begin{table}[ht]
\renewcommand{\arraystretch}{1.25}
\centering
\resizebox{.45\textwidth}{!}{
\begin{tabular}{@{}L{2.1cm}l@{}}
\toprule
Function & Equation \\ 
\midrule
  Softmax & $f(x_i) = \dfrac{\exp{\left(x_i\right)}}{\displaystyle \sum_{j \in K} \exp{\left(x_j\right)}}$ \\[2.5em]
  Sigmoid & $f(x) = \dfrac{1}{1+\exp{\left(-x\right)}}$\\[2.5em]
  Hyperbolic Tangent      & $f(x) = \dfrac{\exp{\left(x\right)} - \exp{\left(-x\right)}}{\exp{\left(x\right)} + \exp{\left(-x\right)}}$\\[2em]
  Cross Entropy Loss & $    L(y,\hat{y}) = -\dfrac{1}{\eta}\displaystyle\sum_{i=1}^K y_{i} \ln\left({\hat{y}_{i}}\right)$\\[2em]
  Weighted Cross Entropy Loss & $    L(\beta,y,\hat{y}) = -\dfrac{1}{\eta}\displaystyle \sum_{i=1}^K \beta_i y_{i} \ln\left({\hat{y}_{i}}\right)$\\
\botrule
\end{tabular}}
\caption{Activation and loss functions for training networks. $\eta$ is the number of samples, $K$ is total number of classes, $\beta_i$ is the weight for class $i$, $y_i$ is the actual value for class $i$, and $\hat{y}_i$ is the prediction value for class $i$.}
\label{tab:activationloss}
\end{table}

\section{Example of Sparse Partial Matrix Completion}\label{secA2}
Let $P$ be the price matrix where the rows are time steps and the columns are assets.
\begin{align}
\textbf{P} = 
\begin{bmatrix}
    2 & 3 & 5 & 2\\
    6 & 7 & 9 & 3\\
    4 & 8 & 6 & 5\\
    5 & 2 & 1 & 2\\
    2 & 5 & 3 & 6\\
\end{bmatrix}
\end{align}
 We compute the return matrix $X$ as in Equation 2.
 \begin{align}
X = 
\begin{bmatrix}
    2 & 1.3 & 0.8 & 0.5\\
    -0.3 & 0.14 & -0.3 & 0.67\\
    0.25 & -0.75 & -0.83 & -0.6\\
    -0.6 & 1.5 & 2 &  2\\
\end{bmatrix}
\end{align}
This gives us the covariance $cov(X) = \Sigma$ and correlation $corr(X) = \theta$. 
\begin{align}
    \Sigma = 
    \begin{bmatrix}
    1.37  & 0.27 & -0.08 &  -0.47\\
    0.27 & 1.12 & 1.25 & 0.93 \\
    -0.08 & 1.25 & 1.59 & 1.21\\
    -0.47 & 0.93 & 1.21 &  1.14 \\
\end{bmatrix}\\
    \theta = 
    \begin{bmatrix}
    1  & 0.21 & -0.05 &  -0.38 \\
    0.21 & 1 & 0.93 & 0.82 \\
    -0.05 & 0.93 & 1 & 0.9 \\
    -0.38 & 0.82 & 0.9 & 1 \\
\end{bmatrix}
\end{align}
Sorting the elements of $\theta$ in order of magnitude, we get
\begin{align}
    \tau = 
    \{0.05, 0.21, 0.38, 0.82, 0.9, 0.93, 1\}
\end{align}
We start with $\tau_{(1)} = 0.05$ and obtain the following matrix.
\begin{align}
   &\hat{ \Sigma}(\tau_{(1)}) = 
    \begin{bmatrix}
    1.37  & 0.27 & 0&  -0.47\\
    0.27 & 1.12 & 1.25 & 0.93 \\
    0 & 1.25 & 1.59 & 1.21\\
    -0.47 & 0.93 & 1.21 &  1.14 \\
\end{bmatrix}
    \not\succeq 0 
\end{align}
And following the same process for $\tau_{(2:n)}$
\begin{align}
&\hat{ \Sigma}(\tau_{(2)}) = 
    \begin{bmatrix}
    1.37  & 0 & 0&  -0.47\\
   0 & 1.12 & 1.25 & 0.93 \\
    0 & 1.25 & 1.59 & 1.21\\
    -0.47 & 0.93 & 1.21 &  1.14 \\
\end{bmatrix}
 \succeq 0 \\
& \hat{ \Sigma}(\tau_{(3)}) = 
    \begin{bmatrix}
    1.37  & 0 & 0&  0\\
   0 & 1.12 & 1.25 & 0.93 \\
    0 & 1.25 & 1.59 & 1.21\\
    0 & 0.93 & 1.21 &  1.14 \\
\end{bmatrix}
\succeq 0 \\
&\hat{ \Sigma}(\tau_{(4)}) = 
    \begin{bmatrix}
    1.37  & 0 & 0&  0\\
   0 & 1.12 & 1.25 & 0 \\
    0 & 1.25 & 1.59 & 1.21\\
    0 & 0 & 1.21 &  1.14 \\
\end{bmatrix}
\not\succeq 0 \\
& \hat{ \Sigma}(\tau_{(5)}) = 
        \begin{bmatrix}
    1.37  & 0 & 0&  0\\
   0 & 1.12 & 1.25 & 0 \\
    0 & 1.25 & 1.59 & 0\\
    0 & 0 & 0 &  1.14 \\
\end{bmatrix}
\succeq 0 \\
&\hat{ \Sigma}(\tau_{(6)}) = 
    \begin{bmatrix}
    1.37  & 0 & 0 &  0\\
    0 & 1.12 & 0 & 0 \\
    0 & 0 & 1.59 & 0 \\
    0  & 0 & 0 &  1.14 \\
\end{bmatrix}
\end{align}
This method gives $ \{\tau_{2},\tau_3,\tau_5,\tau_6,\tau_7\} \in \hat{\tau} $. See that the matrix became indefinite in the process but the final matrix is positive semidefinite. As mentioned, there is no pattern other than $\tau_{(n)} \in \hat{\tau}$.

Now let's sparsify with the partial matrix completion method for the $\tau_{(i)} \notin \hat{\tau}$.
\begin{align*}
   \hat{ \Sigma}(\tau_{(1)}) =& 
    \begin{bmatrix}
    1.37  & 0.27 & 0&  -0.47\\
    0.27 & 1.12 & 1.25 & 0.93 \\
    0 & 1.25 & 1.59 & 1.21\\
    -0.47 & 0.93 & 1.21 &  1.14 \\
\end{bmatrix}
\not\succeq \\
\implies 
&\begin{bmatrix}
    1.37  & 0.27 &-0.08 &  -0.47\\
    0.27 & 1.12 & 1.25 & 0.93 \\
    -0.08 & 1.25 & 1.59 & 1.21\\
    -0.47 & 0.93 & 1.21 &  1.14 \\
\end{bmatrix}
\succeq 0
\\[1em]
 \hat{ \Sigma}(\tau_{(4)}) =& 
    \begin{bmatrix}
    1.37  & 0 & 0&  0\\
   0 & 1.12 & 1.25 & 0 \\
    0 & 1.25 & 1.59 & 1.21\\
    0 & 0 & 1.21 &  1.14 \\
\end{bmatrix}
\not\succeq 0 \\ 
\implies 
&\begin{bmatrix}
    1.37  & 0& 0&  0\\
   0 & 1.12 & 1.25 & 0.93 \\
    0 & 1.25 & 1.59 & 1.21\\
    0 & 0.93 & 1.21 &  1.14 \\
\end{bmatrix}
\succeq 0
\end{align*}
For a less sparse matrix, such as $\hat{ \Sigma}(\tau_{(1)})$, the matrix completion method loses the sparsity. However, we see that for $\hat{ \Sigma}(\tau_{(4)}) $, we are able to gain back positive semidefiniteness by adding in only 2 elements. The trade-off of sparsity and positive semidefiniteness is shown in Figure \ref{fig:matrixComp}.




\end{appendices}


\bibliographystyle{sn-aps}
\bibliography{ref}

\end{document}